\documentclass[paper]{ieice}

\usepackage{amsmath,amssymb}
\usepackage[dvipdfm]{graphicx}
\usepackage{theorem}

\field{D}
\title[SDP-Method based on Strong Computability for Higher-Order Rewrite Systems]{%
  Static Dependency Pair Method
  based on\\ Strong Computability
  for Higher-Order Rewrite Systems}
\authorlist{
  \authorentry[kusakari@is.nagoya-u.ac.jp]{Keiichirou KUSAKARI}{m}{NAGOYA}
  \authorentry[isogai@trs.cm.is.nagoya-u.ac.jp]{Yasuo ISOGAI}{n}{NAGOYA}
  \authorentry[sakai@is.nagoya-u.ac.jp]{Masahiko SAKAI}{m}{NAGOYA}
  \authorentry[frederic.blanqui@inria.fr]{Fr\'{e}d\'{e}ric Blanqui}{n}{FRANCE}
}
\affiliate[NAGOYA]{%
  Graduate School of Information Science, Nagoya Univ.}
\affiliate[FRANCE]{%
  INRIA \& LORIA, France}

\received{2008}{3}{0}
\revised{2008}{6}{0}
\finalreceived{2008}{0}{0}

\def\imply{\Rightarrow}
\def\gsim{\gtrsim}

\def\da{{\mathcal{\downarrow}}}

\def\B{{\cal B}}
\def\C{{\cal C}}
\def\D{{\cal D}}
\def\S{{\cal S}}
\def\T{{\cal T}}
\def\V{{\cal V}}
\def\SC{\scriptscriptstyle SC}

\def\args{\mathit{args}}
\def\safe{\mathit{safe}}

\newcommand{\pair}[1]{\langle#1\rangle}
\newcommand{\sym}[1]{\mathtt{#1}}
\newcommand{\ol}[1]{\overline{#1}}
\newcommand{\ul}[1]{\underline{#1}}

\newcommand{\Reduction}[9]{\mathrel{%
  {}_{#1}^{#2}%
  \raise#7\hbox{%
    \vtop{\ialign{##\crcr%
      \hfil\raise#8\hbox{$\scriptstyle{\ #4\ }$}\hfil\crcr%
      \noalign{\nointerlineskip}%
      #9\crcr%
      \noalign{\nointerlineskip}%
      $\hfil\scriptscriptstyle{\ #3\ }\hfil$\crcr%
    }}%
  }%
  {}_{\scriptscriptstyle #5}^{#6}}}%

\newcommand{\red}[2]{%
  \Reduction{}{}{#1}{#2}{}{}{0.8ex}{0.2ex}{\rightarrowfill}}

\theorembodyfont{\normalfont}
\newtheorem{definition}{Definition}[section]
\newtheorem{lemma}[definition]{Lemma}
\newtheorem{theorem}[definition]{Theorem}
\newtheorem{corollary}[definition]{Corollary}
\newtheorem{proposition}[definition]{Proposition}
\newtheorem{example}[definition]{Example}
\newtheorem{proof}{Proof.\hspace*{-3pt}}

\def\qed{\hspace*{\fill}$\square$}

\begin{document}
\maketitle

\begin{summary}
Higher-order rewrite systems (HRSs)
and simply-typed term rewriting systems (STRSs)
are computational models of functional programs.
We recently proposed an extremely powerful method,
the static dependency pair method,
which is based on the notion of strong computability,
in order to prove termination in STRSs.
In this paper, we extend the method to HRSs.
Since HRSs include $\lambda$-abstraction but STRSs do not,
we restructure the static dependency pair method
to allow $\lambda$-abstraction,
and show that the static dependency pair method also works well on HRSs
without new restrictions.
\end{summary}


\section{Introduction}

A term rewriting system (TRS) is a computational model
that provides operational semantics for functional programs \cite{T03}.
A TRS cannot, however, directly handle higher-order functions,
which are widely used in functional programming languages.
Simply-typed term rewriting systems (STRSs) \cite{K01}
and higher-order rewrite systems (HRSs) \cite{N91}
have been introduced to extend TRSs.
These rewriting systems can directly handle higher-order functions.
For example, a typical higher-order function $\sym{foldl}$
can be represented by the following HRS $R_{\sym{foldl}}$:
\[\left\{\begin{array}{l}
  \sym{foldl}(\lambda xy.F(x,y), X, \sym{nil}) \to X \\
  \sym{foldl}(\lambda xy.F(x,y), X, \sym{cons}(Y,L)) \\
    \hspace*{70pt}\to \sym{foldl}(\lambda xy.F(x,y), F(X,Y), L)
\end{array}\right.\]
HRSs can represent anonymous functions
because 
HRSs have a $\lambda$-abstraction syntax,
which STRSs do not.
For instance,
an anonymous function $\lambda xy. \sym{add}(x,\sym{mul}(y,y))$ is used
in the HRS $R_{\sym{sqsum}}$,
which is the union of
$R_{\sym{foldl}}$ and the following rules:
\[\left\{\begin{array}{l}
  \sym{add}(0,Y)    \to Y           \\
  \sym{add}(\sym{s}(X),Y) \to \sym{s}(\sym{add}(X,Y)) \\
  \sym{mul}(0,Y)    \to 0           \\
  \sym{mul}(\sym{s}(X),Y) \to \sym{add}(\sym{mul}(X,Y),Y) \\
  \sym{sqsum}(L)     \to \sym{foldl}(\lambda xy.\sym{add}(x,\sym{mul}(y,y)), 0, L)
\end{array}\right.\]
Here, the function $\sym{sqsum}$ returns
the square sum
$x_1^2 + x_2^2 + \cdots + x_n^2$
from an input list $[x_1, x_2, \dots, x_n]$.

As a method for proving termination of TRSs,
Arts and Giesl proposed the dependency pair method for TRSs
based on recursive structure analysis \cite{AG00},
which was then extended to STRSs \cite{K01}, and to HRSs \cite{SWS01}.

In higher-order settings,
there are two kinds of analysis for recursive structures.
One is dynamic analysis, and the other is static analysis.
The extensions in \cite{K01} and \cite{SWS01} analyze
dynamic recursive structures
based on function-call dependency relationships,
but not on relationships that may be extracted syntactically
from function definitions.
When a program runs,
some functions can be substituted for higher-order variables.
Dynamic recursive structure analysis
considers dependencies through higher-order variables.
Static recursive structure analysis on the other hand,
does not consider such dependencies.

For example, consider the HRS $R_{\sym{sqsum}}$.
The dynamic dependency pair method in \cite{SWS01}
extracts the following 9 pairs,
called dynamic dependency pairs:
\[\left\{\begin{array}{lr}
  \sym{foldl}^\sharp(\lambda xy.F(x,y), X, \sym{cons}(Y,L))
    \\\hspace*{50pt}
    \to \sym{foldl}^\sharp(\lambda xy.F(x,y), F(X,Y), L)
    & (a) \\
  \sym{foldl}^\sharp(\lambda xy.F(x,y), X, \sym{cons}(Y,L))
    \to F(c_x,c_y)
    & (b) \\
  \sym{foldl}^\sharp(\lambda xy.F(x,y), X, \sym{cons}(Y,L))
    \to F(X,Y)
    & (c) \\
  \sym{add}^\sharp(\sym{s}(X),Y) \to \sym{add}^\sharp(X,Y)
    & (d) \\
  \sym{mul}^\sharp(\sym{s}(X),Y) \to \sym{add}^\sharp(\sym{mul}(X,Y),Y)
    & (e) \\
  \sym{mul}^\sharp(\sym{s}(X),Y) \to \sym{mul}^\sharp(X,Y)
    & (f) \\
  \sym{sqsum}^\sharp(L)
    \to \sym{foldl}^\sharp(\lambda xy.\sym{add}(x,\sym{mul}(y,y)), 0, L)
    & (g) \\
  \sym{sqsum}^\sharp(L)
    \to \sym{add}^\sharp(c_x,\sym{mul}(c_y,c_y))
    & (h) \\
  \sym{sqsum}^\sharp(L)
    \to \sym{mul}^\sharp(c_y,c_y)
    & (i)
\end{array}\right.\]
Here $c_x,c_y$ are fresh constants
corresponding to the bound variables $x$ and $y$.
The dynamic dependency pair method returns the following 15 components,
called dynamic recursion components:
\[\left\{\begin{array}{l}
    \{ (a) \}, \{ (b) \}, \{ (c) \}, \{ (d) \}, \{ (f) \},
    \{ (a),(b) \},\\
    \{ (a),(c) \}, \{ (b),(c) \}, \{ (b),(g) \}, \{ (c),(g) \},\\
    \{ (a),(b),(c) \}, \{ (a),(b),(g) \}, \{ (a),(c),(g) \},\\
    \{ (b),(c),(g) \},
    \{ (a),(b),(c),(g) \}
\end{array}\right\}\]
It is intuitive that
this recursive structure analysis may be unnatural and intractable.
The problem is caused by function-call dependency relationships
through the higher-order variable $F$.

The static dependency pair method,
which is based on definition dependency relationships,
can solve the unnatural and intractable problem above.
Since the static dependency pair method can ignore terms
headed by a higher-order variable which are difficult to handle,
in this meaning
the static dependency pair method is more natural and more powerful
than the dynamic dependency pair method.
In fact, the static dependency pair method presented in this paper
shows that $R_{\sym{sqsum}}$ only has
the following 3 static recursion components:
\[\begin{array}{l}
  \left\{\begin{array}{l}
    \sym{foldl}^\sharp(\lambda xy.F(x,y), X, \sym{cons}(Y,L)) \\
    \hspace*{60pt}\to \sym{foldl}^\sharp(\lambda xy.F(x,y), F(X,Y), L)
  \end{array}\right.\\
  \left\{\begin{array}{l}
    \sym{add}^\sharp(\sym{s}(X),Y) \to \sym{add}^\sharp(X,Y)
  \end{array}\right.\\
  \left\{\begin{array}{l}
    \sym{mul}^\sharp(\sym{s}(X),Y) \to \sym{mul}^\sharp(X,Y)
  \end{array}\right.
\end{array}\]

The first result for the static dependency pair method
was given by Sakai and Kusakari \cite{SK05}.
However, this result demanded that target HRSs be
either `strongly linear' or `non-nested',
which is a very strong restriction.
By reconstructing a dependency pair method
based on the notion of strong computability,
Kusakari and Sakai proposed the static dependency pair method for STRSs
and showed that the method is sound
for plain function-passing STRSs \cite{KS07}.
Note that strong computability
was introduced for proving termination in typed $\lambda$-calculus,
which is a stronger condition
than the property of termination \cite{G72,T67}.
`Plain function-passing' means that every higher-order variable
occurs in an argument position on the left-hand side.
Since many non-artificial functional programs are plain function-passing,
this method has a general versatility.
In this paper,
we extend the static dependency pair method 
and the notion of plain function-passing to HRSs.
Since the difference between STRSs and HRSs is
the existence of anonymous functions
(i.e. $\lambda$-abstraction),
extension is necessary.
We show that
our static dependency pair method works well
on plain function-passing HRSs
without new restrictions.

When proving termination by dependency pair methods,
non-loopingness should be shown for each recursion component.
The notion of the subterm criterion \cite{HM04} is frequently utilized,
as is that of a reduction pair \cite{KNT99},
which is an abstraction of the weak-reduction order \cite{AG00}.
The subterm criterion was slightly improved
by extending the subterms permitted by the criterion \cite{KS07}.
Since the subterm criterion and reduction pairs
are effective in termination proofs,
we also reformulate these notions for HRSs.
An effective and efficient method of proving termination
in plain function-passing HRSs is obtained as a result.
These results can be used to prove the termination of $R_{\sym{sqsum}}$,
which cannot be achieved with the dynamic dependency pair method in \cite{SWS01}.
It can easily be seen that
each static recursion component satisfies the subterm criterion
in the underlined positions:
\[\begin{array}{l}
  \left\{\begin{array}{l}
    \sym{foldl}^\sharp(\lambda xy.F(x,y), X, \ul{\sym{cons}(Y,L)}) \\
    \hspace*{60pt}\to \sym{foldl}^\sharp(\lambda xy.F(x,y), F(X,Y), \ul{L})
  \end{array}\right.\\
  \left\{\begin{array}{l}
    \sym{add}^\sharp(\ul{\sym{s}(X)},Y) \to \sym{add}^\sharp(\ul{X},Y)
  \end{array}\right.\\
  \left\{\begin{array}{l}
    \sym{mul}^\sharp(\ul{\sym{s}(X)},Y) \to \sym{mul}^\sharp(\ul{X},Y)
  \end{array}\right.
\end{array}\]
The termination of $R_{\sym{sqsum}}$ can thus be shown easily.

The remainder of this paper is organized as follows.
The next section provides preliminaries required later in the paper.
In Section 3,
we introduce the notion of strong computability,
which provides a theoretical rationale
for the static dependency pair method.
In Section 4,
we describe the notion of plain function-passing.
In Section 5,
we present the static dependency pair method
for plain function-passing HRSs,
the soundness of which is guaranteed by the notion of strong computability.
In Section 6,
we introduce the notions of the reduction pair and the subterm criterion
in order to prove the non-loopingness of static recursion components.
Concluding remarks are presented in Section 7.

\section{Preliminaries}

In this section, 
we give preliminaries needed later on.
We assume that the reader is familiar with notions
for TRSs and HRSs \cite{T03}.

The set $\S$ of {\em simple types}
is generated from the set $\B$ of {\em basic types}
by the type constructor $\to$.
A {\em functional type} or a {\em higher-order type}
is a simple type of the form $\alpha \to \beta$.
We denote by $\V_{\alpha}$
the set of variables of type $\alpha$,
and denote by $\Sigma_{\alpha}$
the set of function symbols of type $\alpha$.
We define $\V = \bigcup_{\alpha \in \S} \V_\alpha$
and $\Sigma = \bigcup_{\alpha \in \S} \Sigma_\alpha$.
We assume that the sets of variables and function symbols are disjoint.
The set $\T_\alpha^{pre}$ of {\em simply-typed preterms}
with simple type $\alpha$
is generated from sets $\V \cup \Sigma$
by $\lambda$-abstraction and $\lambda$-application.
We denote by $t \da$ the $\eta$-long $\beta$-normal form
of a simply-typed preterm $t$.
The set $\T_\alpha$ of {\em simply-typed terms with a simple type $\alpha$}
is defined as $\{ t \da \mid t \in \T_\alpha^{pre} \}$.
We denote $type(t)=\alpha$ if $t \in T_\alpha$.
We also define the set $\T$ of {\em simply-typed terms}
by $\bigcup_{\alpha \in \S} \T_\alpha$,
and the set $\T_{\B}$ of {\em basic typed terms}
by $\bigcup_{\alpha \in \B} \T_\alpha$.
We write $t^\alpha$ to stand for $t \in \T_\alpha$.
Any term in $\eta$-long $\beta$-normal form is of the form
$\lambda x_1 \cdots x_m. a~t_1~\cdots~t_n$,
where $a$ is a variable or a function symbol.
We remark that
$\lambda x_1 \cdots x_m. a~t_1~\cdots~t_n$
is denoted with
$\lambda x_1 \cdots x_m. a(t_1,\ldots,t_n)$
or $\lambda \ol{x_m}. a(\ol{t_n})$ in short.
The $\alpha$-equality of terms is denoted by $\equiv$.
For a simply-typed term $t \equiv \lambda \ol{x_m}. a(\ol{t_n})$,
the symbol $a$, denoted by $top(t)$,
is said to be the {\em top symbol} of $t$,
and the set $\{t_1,\ldots,t_n\}$, denoted by $\args(t)$,
is said to be {\em arguments} of $t$.
The set of free variables in $t$ denoted by $FV(t)$.
We assume for convenience
that bound variables in a term are all different,
and are disjoint from free variables.
We define the set $Sub(t)$ of {\em subterms} of $t$
by $\{t\} \cup Sub(s)$ if $t \equiv \lambda x. s$;
$\{t\} \cup \bigcup_{i=1}^n Sub(t_i)$ if $t \equiv a(t_1,\ldots,t_n)$.
We use $t \geq_{sub} s$ to represent $s \in Sub(t)$,
and define
$t >_{sub} s$ by $t \geq_{sub} s$ and $t \not\equiv s$.
The set of {\em positions} of a term $t$
is the set $Pos(t)$ of strings over positive integers,
which is inductively defined as
$Pos(\lambda x. t) = \{ \varepsilon \} \cup \{ 1 p \mid p \in Pos(t) \}$
and
$Pos(a(t_1,\ldots,t_n))$
  $= \{\varepsilon\} \cup \bigcup_{i=1}^n \{ ip \mid p \in Pos(t_i)\}$.
The {\em prefix order} $\prec$ on positions
is defined by $p \prec q$
iff $pw = q$ for some $w$ ($\neq \varepsilon$).
The subterm of $t$ at position $p$ is denoted by $t|_p$.

A term containing a special constant $\square_\alpha$
of type $\alpha$ is called a {\em context},
denoted by $C[\,]$.
We use $C[t]$ for the term obtained from $C[\,]$
by replacing $\square_\alpha$ with $t^\alpha$.
A substitution $\theta$ is a mapping from variables to terms
such that $\theta(X)$ has a same type of $X$ for each variable $X$.
We define $Dom(\theta) = \{ X \mid X \not\equiv \theta(X) \}$.
A substitution is naturally extended to a mapping from terms to terms.

A {\em rewrite rule} is a pair $(l,r)$ of terms,
denoted by $l \to r$,
such that $top(l) \in \Sigma$,
$type(l) = type(r) \in \B$ and $FV(l) \supseteq FV(r)$\footnote{%
  In order to guarantee
  the decidability of higher-order pattern-matching,
  Nipkow restricts rewrite rules by the notion of pattern \cite{N91}.
  Such a restriction, however, is not necessary to our study.}.
A higher-order rewrite system (HRS) is a set of rules.
The {\em reduction relation} $\red{R}{}$ of an HRS $R$
is defined by $s \red{R}{} t$
iff 
$s \equiv C[l \theta \da]$ and $t \equiv C[r \theta \da]$
for some rule $l \to r \in R$,
context $C[\,]$ and substitution $\theta$.
The transitive-reflexive closure of $\red{R}{}$
is denoted by $\red{R}{*}$.

\begin{proposition}\cite{MN98}\label{prop:close-substitution}
If $s \red{R}{*} t$
then $s \theta \da \red{R}{*} t \theta \da$.
\end{proposition}

A term $t$ is said to be
{\em terminating} or {\em strongly normalizing} in an HRS $R$,
denoted by $SN(R,t)$,
if there is no infinite sequence of $R$ steps
starting from $t$.
We simply denote $SN(R)$
if $SN(R,t)$ holds for any term $t$.
We also define $\T_{SN}(R) = \{ t \mid SN(R,t) \}$,
$\T_{\neg SN}(R) = \T \setminus \T_{SN}(R)$,
and $\T_{SN}^{\args}(R) = \{ t \mid \forall u \in \args(t).SN(R,u)\}$.

All top symbols of the left-hand sides of rules in an HRS $R$,
denoted by $\D_R$, are called {\em defined},
whereas all other function symbols,
denoted by $\C_R$, are {\em constructors}.
We define the {\em marked term} $t^\sharp$
by $a^\sharp(t_1,\ldots,t_n)$
if $t$ has a form $a(t_1,\ldots,t_n)$ with $a \in \D_R$;
otherwise $t^\sharp \equiv t$.
Here $a^\sharp$ is called a {\em marked symbol}.

\section{Strong Computability}

In this section,
we define the notion of strong computability,
introduced for proving termination in typed $\lambda$-calculus,
which is a stronger condition
than the property of termination \cite{G72,T67}.
This notion provides a theoretical rationale
for the static dependency pair method.

\begin{definition}[Strong Computability]
A term $t$ is said to be {\em strongly computable} in an HRS $R$
if $SC(R,t)$ holds,
which is inductively defined on simple types as follows:
\begin{itemize}
\item
  in case of $type(t) \in \B$,
  $SC(R, t)$ is defined as $SN(R,t)$,
\item
  in case of $type(t)=\alpha \to \beta$,
  $SC(R, t)$ is defined as
  $\forall u \in \T_\alpha. (SC(R, u) \imply SC(R, (t u)\da))$.
\end{itemize}
We also define $\T_{SC}(R) = \{ t \mid SC(R,t) \}$,
$\T_{\neg SC}(R) = \T \setminus \T_{SC}(R)$,
and $\T_{SC}^{\args}(R) = \{ t \mid \forall u \in \args(t).SC(R,u)\}$.
\end{definition}

Here we give the basic properties for strong computability,
needed later on.

\begin{lemma}\label{lem:sc1}
For any HRS $R$, the following properties hold:
\begin{description}
\item{(1)}
  For any $(t_0\,t_1\,\cdots\,t_n)\da \in \T$,
  if $SC(R, t_i)$ holds for all $t_i$,
  then $SC(R, (t_0\,t_1\,\cdots\,t_n)\da)$.
\item{(2)}
  For any $t^{\alpha_1 \to \cdots \to \alpha_n \to \alpha}$,
  if $\neg SC(R, t)$,
  then there exist strongly computable terms
  $u_i^{\alpha_i}$ ($1 \leq i \leq n$)
  such that $\neg SC(R, (t\,u_1\,\cdots\,u_n)\da)$.
\item{(3)}
  $SC(R, s)$ and $s \red{R}{*} t$ implies $SC(R, t)$,
  for all $s, t$.
\item{(4)}
  The $\eta$-long $\beta$-normal form $z\da$
  of any variable $z^\alpha$
  is strongly computable,
  for all types $\alpha$.
\item{(5)}
  $SC(R,t^\alpha)$ implies $SN(R,t^\alpha)$,
  for all types $\alpha$.
\end{description}
\end{lemma}

\begin{proof}
The properties (1) and (2) are easily shown by induction on $n$.
\begin{description}
\item{(3)}
  We prove the claim by induction on $type(t)$.
  The case $type(t) \in \B$ is trivial.
  Suppose that $type(s) = type(t) = \alpha \to \beta$.
  Let $s \equiv \lambda x. s'$,
  $t \equiv \lambda x. t'$,
  and $u^\alpha$ be an arbitrary strongly computable term.
  Since $type(l) \in \B$ for every $l \to r \in R$,
  we have $s' \red{R}{*} t'$.
  From Proposition \ref{prop:close-substitution},
  we have
  $(s u)\da \equiv s'\{x:=u\} \red{R}{*} t'\{x:=u\} \equiv (t u)\da$.
  Since $(s u)\da$ is strongly computable,
  $SC(R, (t u)\da)$ follows from the induction hypothesis.
  Hence $t$ is strongly computable.
\item{(4,5)}
  We prove claims by simultaneous induction on $\alpha$.
  The case $\alpha \in \B$ is trivial.
  Suppose that
  $\alpha = \alpha_1 \to \cdots \to \alpha_n \to \beta$
  and $\beta \in \B$.

  Induction step of (4):
  Assume that $z\da$ is not strongly computable
  for some $z \in \V_\alpha$.
  From (2),
  there exist strongly computable terms
  $u_1^{\alpha_1}, \ldots, u_n^{\alpha_n}$
  and $(z(u_1,\ldots,u_n))\da \equiv z(u_1,\ldots,u_n)$
  is not strongly computable.
  From the induction hypothesis (5),
  each $u_i$ is terminating, hence so is $z(u_1,\ldots,u_n)$.
  Since $z(u_1,\ldots,u_n)$ is of basic types,
  $z(u_1,\ldots,u_n)$ is strongly computable.
  This is a contradiction.

  Induction step of (5):
  From the induction hypothesis (4),
  $y\da$ is strongly computable for any $y \in \V_{\alpha_1}$,
  hence so is $(ty)\da$.
  From the induction hypothesis (5),
  $(ty)\da$ is terminating, hence so is $t$.
\qed
\end{description}
\end{proof}

\section{Plain Function-Passing}

The static dependency pair method defined in the next section
cannot be applied to HRSs in general.
For example, consider the HRS
\(R =
  \{ \sym{foo}(\sym{bar}(\lambda x.F(x)))
     \to F(\sym{bar}(\lambda x.F(x)))\}\).
Since the defined symbol $\sym{foo}$
does not occur on the right hand side,
no static recursive structure exists.
However, $R$ is not terminating:
\(\sym{foo}(\sym{bar}(\lambda x.\sym{foo}(x)))
  \red{R}{} \sym{foo}(\sym{bar}(\lambda x.\sym{foo}(x)))
  \red{R}{} \cdots\).
The static dependency pair method therefore requires a suitable restriction.
In \cite{SK05},
we introduced the notions of `strongly linear' and `non-nested' HRSs.
However, these restrictions are too tight.
For STRSs
we presented the notion of plain function-passing,
which covers practical level programs \cite{KS07}.
Intuitively, plain function-passing means that
higher-order free variables on the left-hand side
are passed to the right-hand side directly.
In this section,
we extend the notion of plain function-passing to HRSs.

\begin{definition}
Let $R$ be an HRS and $l \to r \in R$.
We define the set $\safe(l)$ of
{\em safe subterms} of $l$ as the following:
\[\args(l) \cup \bigcup_{l' \in \args(l)}
    \hspace*{-5pt}\{ u \in \safe_{\B}(l', FV(l)) \mid FV(l) \supseteq FV(u) \},\]
where
$\safe_{\B}(\lambda \ol{x_m}. a(\ol{t_n}), X)$ is defined as 
$\{a(\ol{t_n})\}$ if $a \in X$;
otherwise $\{a(\ol{t_n})\} \cup \bigcup_{i=1}^n \safe_{\B}(t_i, X)$.
\end{definition}

We note that $\safe(l) \subseteq Sub(l)$
and any $t \in \safe_{\B}(l',FV(l))$ is of basic types.

\begin{example}\label{ex:pfp1}
Consider HRS $R_{\sym{foldl}}$ displayed in the introduction.
Suppose that
\[l \equiv \sym{foldl}(\lambda xy.F(x,y), Y, \sym{cons}(X, L)).\]
For each argument $u \in \args(l)$,
$\safe_{\B}(u,FV(l))$ is the following:
\begin{eqnarray*}
  \safe_{\B}(\lambda xy.F(x,y), FV(l)) &=& \{ F(x,y) \} \\
  \safe_{\B}(Y, FV(l)) &=& \{ Y \} \\
  \safe_{\B}(\sym{cons}(X, L), FV(l)) &=& \{ \sym{cons}(X, L), X, L \}
\end{eqnarray*}
Since $FV(F(x,y)) \not\subseteq FV(l)$,
safe subterms $\safe(l)$ is the following:
\begin{eqnarray*}
  \safe(l)
    &=& \args(l) \cup \{ Y, \sym{cons}(X, L), X, L \} \\
    &=& \{ \lambda xy.F(x,y), Y, \sym{cons}(X, L), X, L \}
\end{eqnarray*}
\end{example}

We prepare a technical lemma
to show the soundness of the static dependency pair method.

\begin{lemma}\label{lem:pfp}
Let $R$ be an HRS, $l \to r \in R$
and $\theta$ be a substitution.
Then
$l \theta \da \in \T^{\args}_{SC}(R)$
implies
$SC(R, s \theta \da)$ for any $s \in \safe(l)$.
\end{lemma}

\begin{proof}
The case $s \in \args(l)$ is trivial
because $s\theta\da \in \args(l\theta\da)$
follows from $top(l) \in \Sigma$.
Suppose that $s \in \safe_{\B}(l', FV(l))$ and $FV(s) \subseteq FV(l)$
for some $l' \in \args(l)$.
Then
we have $SN(R, l' \theta \da)$
from Lemma \ref{lem:sc1}(5).
Since $type(s) \in \B$
from the definition of $\safe_{\B}$,
it suffices to show $SN(R, s \theta \da)$.
We prove by induction on definition of $\safe_{\B}$
that 
$s \in \safe_{\B}(t, FV(l))$ and $SN(R, t \theta \da)$
implies $SN(R, s \theta \da)$,
for all
$t \equiv \lambda x_1 \cdots x_m. a(t_1, \ldots, t_n) \in Sub(l')$.

The case $t \equiv \lambda x_{1} \cdots x_{m}.s$ is trivial
because 
$t \theta \da \equiv \lambda x_{1} \cdots x_{m}.(s \theta \da)$.
Suppose that $s \in safe_{\B}(t_{j}, FV(l))$ for some $j$.
Without loss of generality,
we can assume that $a \notin Dom(\theta)$ because $a \notin FV(l)$.
Then
$t \theta \da \equiv \lambda \ol{x_m}. a(\ol{t_n\theta\da})$.
Hence, $SN(R, t_{j} \theta \da)$ holds.
From the induction hypothesis, we have $SN(R, s \theta \da)$.
\qed
\end{proof}

\begin{definition}[Plain Function-Passing]
An HRS $R$ is said to be
{\em plain function-passing} (PFP)
if for any $l \to r \in R$ and
$Z(r_1,\ldots,r_n) \in Sub(r)$ such that $Z \in FV(r)$,
there exists $k$ ($\leq n$)
such that $Z(r_1,\ldots,r_k)\da \in \safe(l)$.
We often abbreviate
plain function-passing HRS to PFP-HRS.
\end{definition}

\begin{example}
Referencing to Example \ref{ex:pfp1}.
Since $F\da \equiv \lambda xy.F(x,y) \in \safe(l)$,
HRS $R_{\sym{foldl}}$ is PFP.
\end{example}

\begin{example}
Let $R$ be the following non-terminating HRS:
\[\left\{\begin{array}{l}
  \sym{foo}(\sym{bar}(\lambda x.F(x))) \to F(\sym{bar}(\lambda x.F(x)))
\end{array}\right.\]
Then $R$ is not PFP because:
\[F\da \notin \{\sym{bar}(\lambda x.F(x))\}
  = \safe(\sym{foo}(\sym{bar}(\lambda x.F(x)))).\]
\end{example}

\begin{example}
Let $R$ be the following terminating HRS:
\[\left\{\begin{array}{l}
  \sym{mapfun}(\sym{nil_F}, X) \to \sym{nil} \\
  \sym{mapfun}(\sym{cons_F}(\lambda x. F(x), L), X) \\
  ~~~~~~~~~~~~\to \sym{cons}(F(X), \sym{mapfun}(L, X))
\end{array}\right.\]
Then $R$ is not PFP because:
\begin{eqnarray*}
  F\da
    &\notin& \{ \sym{cons_F}(\lambda x.F(x), L), L, X \} \\
    &=& \safe(\sym{mapfun}(\sym{cons_F}(\lambda x.F(x), L), X))
\end{eqnarray*}
\end{example}

In any PFP-HRS $R$,
for any subterm $Z(r_1,\ldots,r_n)$ headed by a higher-order variable
in the right hand side of a rule $l \to r$,
there exists a prefix $Z(r_1,\ldots,r_k)$
such that $Z(r_1,\ldots,r_k)\da \in \safe(l)$.
Thanks to Lemmas \ref{lem:sc1}(1) and \ref{lem:pfp},
this property guarantees that
$Z(r_1,\ldots,r_n)\theta\da$ is strongly computable
whenever $l \theta \da \in \T^{\args}_{SC}(R)$
and $r_i\theta\da \in \T_{SC}(R)$ ($i=1,\ldots,n$).
This beneficial property
eliminates a dependency analysis through higher-order variables
from static recursive structure analysis
(cf. Lemma \ref{lem:dp2}),
and contributes in obtaining the soundness
of the static dependency pair method (cf. Theorem \ref{th:SDP1}).

In the definition of PFP,
the case $n=0$ must be considered.
That is, any first-order variable in $Var(r)$
should belong to $\safe(l)$.
Otherwise Lemma \ref{lem:pfp} does not hold.
For example,
consider the HRS $R = \{\sym{foo}(F(X)) \to X\}$
and the substitution $\theta = \{F := \lambda x. 0\}$.
Then $X$ does not occur
in $\sym{foo}(0) \equiv \sym{foo}(F(X))\theta\da$,
and we must exclude $R$ from plain function-passing.

Note that every first-order rewrite system is plain function-passing.

A termination condition for higher-order rewrite rules
having a specific form of plain function-passing was investigated
under Jouannaud and Okada's general schema \cite{JO91,JO97}.
The restriction that higher-order variables occur as arguments
is weakened by using the notion of computability closure
\cite{BJO02,B06,B07}.
We leave a similar extension of
the present work with computability closure
for the future.

\section{Static Dependency Pair Method}\label{sec:SDP}

In this section
we present the static dependency pair method for PFP-HRSs.
The recursive structures derived by the static dependency pair method
accord with a programmer's intuition.
Since many existing programs are written
so as to terminate,
this method is of benefit in proving that they do indeed terminate.

First,
we describe candidate terms,
improving on the notion of candidate terms in \cite{SWS01}.
Candidate terms are a variant of subterms,
and bound variables never become free in candidate terms.
This feature is useful for showing the soundness of our method
(cf. Lemma \ref{lem:dp2}).

\begin{definition}[Candidate Term]\label{def:cand}
The set of {\em candidate terms}
of $t \equiv \lambda \ol{x_m}.$ $a(\ol{t_n})$,
denoted by $Cand(t)$,
is defined as follows:
\[Cand(t) =
  \{ t \} \cup \bigcup^{n}_{i=1} Cand(\lambda x_1 \cdots x_m. t_i) \]
\end{definition}

We consider the case of $\sym{foo},\sym{bar} \in \D_R$
and $t \equiv \lambda x. \sym{foo}(\sym{bar},x)$.
Then we have
\[Cand(t) =
  \{ \lambda x. \sym{foo}(\sym{bar},x), \lambda x. \sym{bar}, \lambda x. x \}.\]
Note that 
the definition in \cite{SWS01} gave
\(Cand(t) = \{ \sym{foo}(\sym{bar},c_x), \sym{bar} \}\),
where $c_x$ is a fresh constant corresponding to the bound variable $x$.

Next, we introduce the notion of static dependency pairs
by using candidate terms.
This notion forms the basis for the static dependency pair method.

\begin{definition}[Static Dependency Pair]
Let $R$ be an HRS.
A pair $\pair{l^{\sharp},$ $a^{\sharp}(r_{1}, \ldots, r_{n})}$,
denoted by $l^{\sharp} \to a^{\sharp}(r_{1}, \ldots, r_{n})$,
is said to be a {\em static dependency pair} in $R$
if there exists $l \to r \in R$ such that
\begin{itemize}
\item
  $\lambda x_1 \cdots x_m. a(r_1, \ldots, r_n) \in Cand(r)$,
\item
  $a \in \D_R$, and
\item
  $a(r_{1}, \ldots, r_{k}) \da$ $\notin \safe(l)$
  for all $k$ $(\leq n)$.
\end{itemize}
We denote by $SDP(R)$ the set of static dependency pairs in $R$.
\end{definition}

Notice that static dependency pairs have
no terms headed by a higher-order variable
nor terms of a functional type.

\begin{example}\label{ex:sdp}
For the HRS $R_{\sym{sqsum}}$ displayed in the introduction,
the set $SDP(R_{\sym{sqsum}})$ consists of the following seven pairs:
\[\left\{\begin{array}{l}
  \sym{foldl}^\sharp(\lambda xy.F(x,y), X, \sym{cons}(Y,L))\\
    \hspace*{60pt}\to \sym{foldl}^\sharp(\lambda xy.F(x,y), F(X,Y), L) \\
  \sym{add}^\sharp(\sym{s}(X),Y) \to \sym{add}^\sharp(X,Y) \\
  \sym{mul}^\sharp(\sym{s}(X),Y) \to \sym{add}^\sharp(\sym{mul}(X,Y),Y) \\
  \sym{mul}^\sharp(\sym{s}(X),Y) \to \sym{mul}^\sharp(X,Y)        \\
  \sym{sqsum}^\sharp(L)
    \to \sym{foldl}^\sharp(\lambda xy.\sym{add}(x,\sym{mul}(y,y)), 0, L) \\
  \sym{sqsum}^\sharp(L) \to \sym{add}^\sharp(x,\sym{mul}(y,y)) \\
  \sym{sqsum}^\sharp(L) \to \sym{mul}^\sharp(y,y)
\end{array}\right.\]
Notice that
we use the extra variables $x,y$
in the sixth and seventh dependency pairs.
\end{example}

Each static dependency pair expresses
nothing but the local dependency of functions
based on dependency relationships displayed in rules.
To analyze the global dependency of functions,
in other words,
to analyze the static recursive structure,
we introduce notions of
a static dependency chain and a static dependency graph.

\begin{definition}[Static Dependency Chain]
Let $R$ be an HRS.
A sequence
$u_0^\sharp \to v_0^\sharp, u_1^\sharp \to v_1^\sharp, \cdots$
of static dependency pairs in $R$
is said to be a {\em static dependency chain in $R$}
if there exist $\theta_0,\theta_1,\ldots$
such that
$v_i^\sharp\theta_i\da \red{R}{*} u_{i+1}^\sharp\theta_{i+1}\da$
and $u_i\theta_i\da, v_i\theta_i\da \in \T_{SC}^{\args}(R)$
for any $i$.
\end{definition}

\begin{definition}[Static Dependency Graph]
The \\
{\em static dependency graph of $R$} is a directed graph,
in which nodes are $SDP(R)$
and there exists an arc 
from $u^\sharp \to v^\sharp$ to $u'^\sharp \to v'^\sharp$
if $u^\sharp \to v^\sharp, u'^\sharp \to v'^\sharp$
is a static dependency chain.
\end{definition}

\begin{example}\label{ex:sdg}
The static dependency graph
of the HRS $R_{\sym{sqsum}}$ (cf. Example \ref{ex:sdp})
is shown in Fig. \ref{fig:HRSsqsum}.
\end{example}

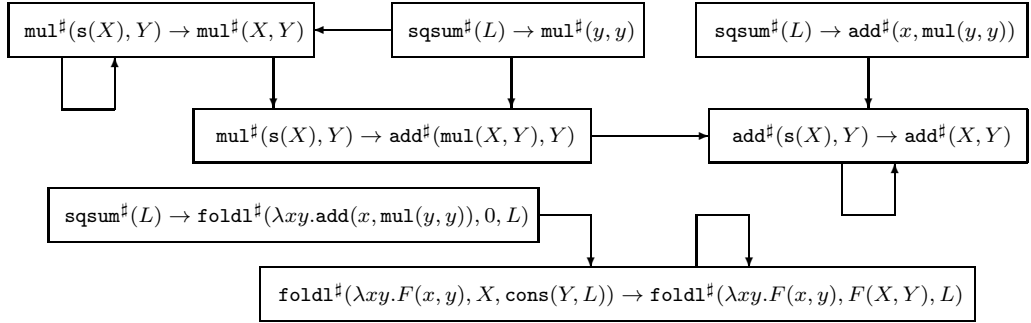
\begin{figure*}[htb]
\begin{center}
\begin{picture}(385,120)(0,0)

\put(0,100){\framebox(115,20){
  $\sym{mul}^\sharp(\sym{s}(X),Y) \to \sym{mul}^\sharp(X,Y)$}}
\put(20,100){\line(0,-1){20}}
\put(20,80){\line(1,0){20}}
\put(40,80){\vector(0,1){20}}
\put(100,100){\vector(0,-1){20}}
\put(145,100){\framebox(95,20){
  $\sym{sqsum}^\sharp(L) \to \sym{mul}^\sharp(y,y)$}}
\put(145,110){\vector(-1,0){30}}
\put(190,100){\vector(0,-1){20}}
\put(260,100){\framebox(125,20){
  $\sym{sqsum}^\sharp(L) \to \sym{add}^\sharp(x,\sym{mul}(y,y))$}}
\put(325,100){\vector(0,-1){20}}

\put(70,60){\framebox(150,20){
  $\sym{mul}^\sharp(\sym{s}(X),Y) \to \sym{add}^\sharp(\sym{mul}(X,Y),Y)$}}
\put(220,70){\vector(1,0){45}}
\put(265,60){\framebox(120,20){
  $\sym{add}^\sharp(\sym{s}(X),Y) \to \sym{add}^\sharp(X,Y)$}}
\put(315,60){\line(0,-1){20}}
\put(315,40){\line(1,0){20}}
\put(335,40){\vector(0,1){20}}

\put(15,30){\framebox(185,20){
  $\sym{sqsum}^\sharp(L)
    \to \sym{foldl}^\sharp(\lambda xy.\sym{add}(x,\sym{mul}(y,y)), 0, L)$}}
\put(200,40){\line(1,0){20}}
\put(220,40){\vector(0,-1){20}}

\put(95,0){\framebox(270,20){
  $\sym{foldl}^\sharp(\lambda xy.F(x,y),X,\sym{cons}(Y,L))
     \to \sym{foldl}^\sharp(\lambda xy.F(x,y),F(X,Y),L)$}}
\put(260,20){\line(0,1){20}}
\put(260,40){\line(1,0){20}}
\put(280,40){\vector(0,-1){20}}

\end{picture}
\end{center}
\caption{static dependency graph of $R_{\sym{sqsum}}$}\label{fig:HRSsqsum}
\end{figure*}

Unfortunately,
the connectability of the static dependency pairs is undecidable.
Hence, we need suitable approximation techniques.
In TRSs,
such techniques were studied \cite{M02}.
One of simple approximated dependency graphs
is the graph in which 
an arc
from $u^\sharp \to v^\sharp$ to $u'^\sharp \to v'^\sharp$
exists 
if $v^\sharp$ and $u'^\sharp$ have the same top symbol.
Note that
for the HRS $R_{\sym{sqsum}}$
this approximation gives the precise static dependency graph
shown in Fig.~\ref{fig:HRSsqsum}.

We now introduce the notions of
static recursion components and non-loopingness.
As usual,
the termination of HRS can be proved
by proving the non-loopingness of each recursion component.
These proofs are similar to the other dependency pair methods.

\begin{definition}[Static Recursion Component]\ \\
Let $R$ be an HRS.
A {\em static recursion component} in $R$
is a set of nodes in a strongly connected subgraph
of the static dependency graph of $R$.
Using $SRC(R)$ we denote the set of static recursion components in $R$.
\end{definition}

\begin{example}\label{ex:src}
The static dependency graph of $R_{\sym{sqsum}}$ (Fig. \ref{fig:HRSsqsum})
has three strongly connected subgraphs.
Thus, the set $SRC(R_{\sym{sqsum}})$ consists of the following three components:
\[\begin{array}{l}
  \left\{\begin{array}{l}
    \sym{foldl}^\sharp(\lambda xy.F(x,y), X, \sym{cons}(Y,L)) \\
    \hspace*{60pt}\to \sym{foldl}^\sharp(\lambda xy.F(x,y), F(X,Y), L)
  \end{array}\right.\\
  \left\{\begin{array}{l}
    \sym{add}^\sharp(\sym{s}(X),Y) \to \sym{add}^\sharp(X,Y)
  \end{array}\right.\\
  \left\{\begin{array}{l}
    \sym{mul}^\sharp(\sym{s}(X),Y) \to \sym{mul}^\sharp(X,Y)
  \end{array}\right.
\end{array}\]
\end{example}

\begin{definition}[Non-Looping]
A static recursion component $C$ in an HRS $R$
is said to be {\em non-looping}
if there exists no infinite static dependency chain
in which only pairs in $C$ occur
and every $u^\sharp \to v^\sharp \in C$ occurs infinitely many times.
\end{definition}

In the remainder of this section,
we show the soundness of the static dependency pair method on PFP-HRSs.
That is, we show that
if any static recursion component of PFP-HRS $R$ is non-looping,
then $R$ is terminating.
We need two lemmas.

\begin{lemma}\label{lem:dp1}
Let $R$ be a non-terminating HRS.
Then $\T_{\B} \cap \T_{\neg SC}(R) \cap \T^{\args}_{SC}(R) \neq \emptyset$.
\end{lemma}

\begin{proof}
Since $R$ is not terminating,
$\T_{\neg\SC}(R) \neq \emptyset$
follows from Lemma \ref{lem:sc1}(5).
Let $t \equiv \lambda x_{1} \cdots x_{m}.a(t_{1}, \ldots, t_{n})$
be a minimal size term in $\T_{\neg\SC}(R)$.
From Lemma \ref{lem:sc1}(2),
there exist $u_{1}, \ldots, u_{m} \in \T_{SC}(R)$
such that $\neg SC(R, t')$
where $t' \equiv (t~u_1~\cdots~u_{m}) \da$.
Suppose that 
$\sigma = \{ x_{j} := u_{j} \mid 1 \leq j \leq m \}$.
Then
$t' \equiv (a\sigma~t_1\sigma~\cdots~t_{n}\sigma) \da$.
Since the size of $t'_{i} \equiv \lambda x_{1} \cdots x_{m}.t_{i}$
is less than the size of $t$,
we have $SC(R, t'_i)$ by the minimality of $t$.
Since
$t_{i} \sigma \da \equiv (t'_{i}~u_{1}~\cdots~u_{m}) \da$,
we have $SC(R, t_{i} \sigma \da)$ by Lemma \ref{lem:sc1}(1).
Assume that $a \in \{ x_{1}, \ldots, x_{m} \}$.
Since $a \sigma \da \equiv u_{j} \in \T_{SC}(R)$,
$SC(R, t')$ follows from Lemma \ref{lem:sc1}(1).
This is a contradiction.
Hence, we have $a \notin \{ x_{1}, \ldots, x_{m} \}$.
Therefore we have
$t' \equiv a(t_{1} \sigma \da, \ldots, t_{n} \sigma \da)$
  $\in \T_{\B} \cap \T_{\neg SC}(R) \cap \T^{\args}_{SC}(R)$.
\qed
\end{proof}

\begin{lemma}\label{lem:dp2}
Let $R$ be a PFP-HRS.
For any
$t \in \T_{\B} \cap \T_{\neg SC}(R) \cap \T^{\args}_{SC}(R)$,
there exist
$l^\sharp \to v^\sharp \in SDP(R)$ and a substitution $\theta$
such that
$t^\sharp \red{R}{*} (l \theta \da)^\sharp$
and
$l\theta\da, v\theta\da
  \in \T_{\B} \cap \T_{\neg SC}(R) \cap \T^{\args}_{SC}(R)$.
\end{lemma}

\begin{proof}
From $t \in \T^{\args}_{SC}(R)$ and Lemma \ref{lem:sc1}(5),
we have $t \in \T^{\args}_{SN}(R)$.
From $t \in \T_{\B} \cap \T_{\neg SC}(R)$,
we have $\neg SN(R, t)$.
Hence, there exist $l \to r \in R$ and a substitution $\theta'$
such that
$t^\sharp \red{R}{*} (l \theta' \da)^\sharp$,
$l \theta' \da, r \theta' \da \in \T_{\neg SN}(R)$,
and $Dom(\theta') \subseteq FV(l)$.
Since $type(l)=type(r) \in \B$,
we have
$l \theta' \da, r \theta' \da \in \T_{\neg SC}(R)$.
Moreover,
$l \theta' \da \in \T^{\args}_{SC}(R)$
follows from Lemma \ref{lem:sc1}(3).
Since $r \in Cand(r)$ and $\neg SC(R, r \theta' \da)$,
we have 
$\{r' \in Cand(r) \mid \neg SC(R, r' \theta' \da)\} \neq \emptyset$.
Let $v' \equiv \lambda x_{1} \cdots x_{m}. a(r_{1}, \ldots, r_{n})$
be a minimal size term in this set.

From Lemma \ref{lem:sc1}(2),
there exist strongly computable terms $u_1,\ldots,u_m$
such that $(v'\theta'~u_1~\cdots~u_m) \da$
is not strongly computable.
Let $v$ and $\theta$ be
$v \equiv a(r_{1}, \ldots, r_{n})$
and 
$\theta = \theta' \cup \{ x_{i} := u_{i} \mid 1 \leq i \leq m \}$.
Since
$v\theta\da \equiv (v'\theta'~u_1~\cdots~u_m) \da$,
we have $v \theta \da \in \T_{\B} \cap \T_{\neg SC}(R)$.
Since
$l \theta \da \equiv l \theta' \da$ from $x_{i} \notin FV(l)$,
we have
$l \theta \da \in \T_{\B} \cap \T_{\neg SC}(R) \cap \T^{\args}_{SC}(R)$.
Since $\lambda x_1 \cdots x_m. r_i \in Cand(r)$,
$SC(R, (\lambda x_1 \cdots x_m. r_i)\theta'\da)$
follows from the minimality of $v'$.
Hence, each
$r_i\theta\da \equiv ((\lambda x_1 \cdots x_m. r_i)\theta'~u_1~\cdots~u_m)\da$
is strongly computable
from Lemma \ref{lem:sc1}(1).

We prove the remaining claims that
$v \theta \da \in \T^{\args}_{SC}(R)$
and $l^\sharp \to v^\sharp \in SDP(R)$.
\begin{itemize}
\item
  Assume that $a \in \{ x_{i} \mid 1 \leq i \leq m \}$.
  Then $SC(R, v \theta \da)$
  follows from $SC(R, a \theta \da)$ and Lemma \ref{lem:sc1}(1).
  This is a contradiction.
\item
  Assume that $a \in FV(r)$.
  Since $R$ is PFP,
  there exists $k$ $(\leq n)$ such that
  $a(r_{1}, \ldots, r_{k}) \da \in \safe(l)$.
  From Lemma \ref{lem:pfp},
  $SC(R, a(r_{1}, \ldots, r_{k}) \theta \da)$ holds.
  From Lemma \ref{lem:sc1}(1), $SC(R, v \theta \da)$ holds.
  This is a contradiction.
\item
  Assume that $a \in \C_{R}$.
  Then $\forall i. SN(R, r_{i} \theta \da)$
  follows from Lemma \ref{lem:sc1}(5).
  From $a \in \C_{R}$, we have $SN(R, v \theta \da)$.
  From $v \in \T_{\B}$, we have $SC(R, v \theta \da)$.
  This is a contradiction.
\item
  Assume that $a \in \D_R$
  and there exists $k$ $(\leq n)$ such that
  $a(r_{1}, \ldots, r_{k}) \da$ $\in \safe(l)$.
  From Lemma \ref{lem:pfp},
  $SC(R, a(r_{1}, \ldots, r_{k}) \theta \da)$ holds.
  From Lemma \ref{lem:sc1}(1), $SC(R, v \theta \da)$ holds.
  This is a contradiction.
\end{itemize}
As shown above,
we have
$a \in \D_R$ and
$a(r_{1}, \ldots, r_{k}) \da \notin \safe(l)$
for all $k$ $(\leq n)$.
Hence $l^\sharp \to v^\sharp \in SDP(R)$.
Moreover,
$v \theta \da \in \T^{\args}_{SC}(R)$ holds
because
$v \theta \da \equiv a(r_{1} \theta \da, \ldots, r_{n} \theta \da)$
and $SC(R, r_{i} \theta \da)$ for any $i$.
\qed
\end{proof}

By using the two lemmas above,
we can show the soundness of the static dependency pair method.

\begin{theorem}\label{th:SDP1}
Let $R$ be a PFP-HRS.
If there exists no infinite static dependency chain
then $R$ is terminating.
\end{theorem}

\begin{proof}
Assume that $\neg SN(R)$.
From Lemma \ref{lem:dp1},
there exists 
$t \in \T_{\B}$ $\cap \T_{\neg SC}(R)$ $\cap \T^{\args}_{SC}(R)$.
By applying Lemma \ref{lem:dp2} repeatedly,
we obtain an infinite static dependency chain,
which leads to a contradiction.
\qed
\end{proof}

\begin{corollary}\label{co:SDP1}
Let $R$ be a PFP-HRS
such that there exists no infinite path\footnote{%
  Each node cannot appear more than once in a {\em path}.}
in the static dependency graph.
If all static recursion components are non-looping,
then $R$ is terminating.
\end{corollary}

Note that no infinite path condition in this corollary
is always satisfied for finite PFP-HRSs,
since nodes are finite in the static dependency graph.

\section{Non-Loopingness}

In section \ref{sec:SDP}
we showed that a PFP-HRS terminates
if every static recursion component is non-looping.
In order to show non-loopingness,
the notion of the subterm criterion \cite{HM04,KS07} is frequently utilized,
as is that of a reduction pair \cite{KNT99},
which is an abstraction of the weak-reduction order\footnote{%
  A quasi-order $\gsim$ is said to be a {\em weak reduction order}
  if the pair $(\gsim,\gnsim)$
  of $\gsim$ and its strict part $\gnsim$
  is a reduction pair.}\cite{AG00}.
These techniques are also effective in termination proofs for HRSs.
We begin with reduction pairs.

\begin{definition}[Reduction Pair]
Let $\gsim$ be a quasi-order and $>$ be a strict order.
The pair $(\gsim,>)$ is said to be a {\em reduction pair}
if the following properties hold:
\begin{itemize}
\item
  $>$ is well-founded and closed under substitution,
\item
  $\gsim$ is closed under contexts and substitutions, and
\item
  $\mathord{\gsim} \cdot \mathord{>} \subseteq \mathord{>}$
  or $\mathord{>} \cdot \mathord{\gsim} \subseteq \mathord{>}$.
\end{itemize}
\end{definition}

\begin{lemma}\label{lem:RP}
Let $R$ be an HRS and $C \in SRC(R)$.
If there exists a reduction pair $(\gsim, >)$
such that 
$R \subseteq \mathord{\gsim}$,
$C \subseteq \mathord{\gsim} \cup \mathord{>}$,
and $C \cap \mathord{>} \neq \emptyset$,
then $C$ is non-looping.
\end{lemma}

\begin{proof}
Obvious.
\qed
\end{proof}

Next we introduce the subterm criterion for HRSs.
In \cite{HM04},
Hirokawa and Middeldorp proved that
the subterm criterion guarantees the non-loopingness in TRSs.
The key of the proof is that
the relation $\red{R}{} \cup >_{sub}$ is well-founded
on terminating terms.
Since the property also holds in higher-order rewriting,
we directly ported the criterion to STRSs \cite{KS07}.
We also slightly improved the subterm criterion
by extending the codomain of a function $\pi$
from positive integers
to sequences of positive integers \cite{KS07}.
In the following,
we extend the improved subterm criterion onto HRSs,
that is to handle $\lambda$-abstraction.

\begin{definition}[Subterm Criterion]
Let $R$ be an HRS and $C \in SRC(R)$.
We say that $C$ satisfies the {\em subterm criterion}
if there exists a function $\pi$
from $\D_R$ to non-empty sequences of positive integers such that
\begin{description}
\item{($\alpha$)}
  $u|_{\pi(top(u))} >_{sub} v|_{\pi(top(v))}$
  for some $u^\sharp \to v^\sharp \in C$, and
\item{($\beta$)}
  the following conditions hold for any $u^\sharp \to v^\sharp \in C$:
  \begin{itemize}
  \item
    $u|_{\pi(top(u))} \geq_{sub} v|_{\pi(top(v))}$,
  \item
    $\forall p \prec \pi(top(u)). top(u|_{p}) \notin FV(u)$, and
  \item
    $\forall q \prec \pi(top(v)). q = \varepsilon \lor top(v|_q) \notin FV(v) \cup \D_R$.
  \end{itemize}
\end{description}
\end{definition}

\begin{lemma}\label{lem:subterm}
Let $R$ be an HRS and $C \in SRC(R)$.
If $C$ satisfies the subterm criterion
then $C$ is non-looping.
\end{lemma}

\begin{proof}
Assume that pairs in $C$ generate an infinite chain
\(u_0^\sharp \to v_0^\sharp, u_1^\sharp \to v_1^\sharp, u_2^\sharp \to v_2^\sharp, \cdots\)
in which every $u^\sharp \to v^\sharp \in C$ occurs
infinitely many times,
and let $\theta_0,\theta_1,\ldots$ be substitutions
such that
$v_i^\sharp\theta_i\da \red{R}{*} u_{i+1}^\sharp\theta_{i+1}\da$
and $u_i\theta_i\da, v_i\theta_i\da \in \T_{SC}^{\args}(R)$
for each $i$.
From Lemma \ref{lem:sc1}(5),
$u_i\theta_i\da, v_i\theta_i\da \in \T_{SN}^{\args}(R)$.
Denote $\pi(top(u_i))$ by $p_i$ for each $i$.
Since $v_i^\sharp\theta_i\da \red{R}{*} u_{i+1}^\sharp\theta_{i+1}\da$,
we have $top(v_i)=top(u_{i+1})$.
Hence, from the condition ($\beta$) of the subterm criterion, we have
\[(u_0\theta_0\da)|_{p_0} \geq_{sub} (v_0\theta_0\da)|_{p_1}
  \red{R}{*} (u_1\theta_1\da)|_{p_1} \geq_{sub} \cdots. \]
From the condition ($\alpha$) of the subterm criterion,
the sequence above contains infinitely many $>_{sub}$.
Hence
there exists an infinite sequence
starting with $(u_0\theta_0\da)|_j$
with respect to $\red{R}{} \cup >_{sub}$,
where $j$ is the positive integer such that $j \preceq p_0$.
This is a contradiction with
$u_0\theta_0\da \in \T_{SN}^{\args}(R)$.
\qed
\end{proof}

Finally,
we present a powerful method for proving termination of PFP-HRSs.

\begin{theorem}\label{th:SDP2}
Let $R$ be a PFP-HRS
such that there exists no infinite path
in the static dependency graph.
If any static recursion component $C \in SRC(R)$
satisfies one of the following properties,
then $R$ is terminating.
\begin{itemize}
\item
  $C$ satisfies the subterm criterion.
\item
  There exists a reduction pair $(\gsim,>)$
  such that
  $R \subseteq \mathord{\gsim}$,
  $C \subseteq \mathord{\gsim} \cup \mathord{>}$,
  and $C \cap \mathord{>} \neq \emptyset$.
\end{itemize}
\end{theorem}

\begin{proof}
From Corollary \ref{co:SDP1} and Lemma \ref{lem:RP}, \ref{lem:subterm}.
\qed
\end{proof}

As seen in the theorem,
proving non-loopingness by the subterm criterion
depends only on a recursion component,
unlike proving one by a reduction pair.
Thus
the approach by the subterm criterion
is more efficient than the approach by reduction pairs.

\begin{example}
We show the termination of PFP-HRS $R_{\sym{sqsum}}$
displayed in the introduction.
Let $\pi(\sym{foldl}) = 3$, $\pi(\sym{add}) = 1$, and $\pi(\sym{mul}) = 1$.
Then
all $C \in SRC(R_{\sym{sqsum}})$ (cf. Example \ref{ex:src})
satisfy the subterm criterion
in the underlined positions below:
\[\begin{array}{l}
  \left\{\begin{array}{l}
    \sym{foldl}^\sharp(\lambda xy.F(x,y), X, \ul{\sym{cons}(Y,L)}) \\
    \hspace*{60pt}\to \sym{foldl}^\sharp(\lambda xy.F(x,y), F(X,Y), \ul{L})
  \end{array}\right.\\
  \left\{\begin{array}{l}
    \sym{add}^\sharp(\ul{\sym{s}(X)},Y) \to \sym{add}^\sharp(\ul{X},Y)
  \end{array}\right.\\
  \left\{\begin{array}{l}
    \sym{mul}^\sharp(\ul{\sym{s}(X)},Y) \to \sym{mul}^\sharp(\ul{X},Y)
  \end{array}\right.
\end{array}\]
Hence the termination can be shown by Theorem \ref{th:SDP2}.
\end{example}

\section{Concluding Remarks}

In this paper,
we extended the static dependency pair method
based on strong computability for STRSs \cite{KS07} to that for HRSs.
The following topics remain for future work.

\begin{itemize}
\item
{\em Argument filtering method for HRSs}:
Since it is generally difficult to design reduction pairs,
the argument filtering method
was proposed for the dependency pair method of TRSs \cite{AG00},
and extended to STRSs \cite{K01}.
However, there is no known
argument filtering method for HRSs.
The argument filtering method in \cite{K01}
can only be applied to left-firmness systems,
in which every variable of the left-hand sides occurs at a leaf position.
It may be possible to adapt the argument filtering method
for HRSs without the left-firmness restriction
because the counterexample shown in \cite{K01}
is no longer a counterexample for HRSs.

\item
{\em Notion of usable rules for HRSs}:
The notion of usable rules was introduced for TRSs
by Hirokawa and Middeldorp \cite{HM04},
and by Thiemann, Giesl, and Schneider-Kamp \cite{TGS04}
to reduce constraints
when trying to prove non-loopingness by means of reduction pairs.
These proofs are based on Urbain's proof
of an incremental approach to the dependency pair method \cite{U04}.
It will be of benefit to develop the notion of usable rules for HRSs.

\item
{\em Extending upon the class of plain function-passing}:
We have only shown the soundness of the static dependency pair method
for the class of plain function-passing systems.
The notions of pattern computable closure \cite{B06}
and safe function-passing \cite{KS}
are promising techniques by which this may be extended.
\end{itemize}

\section*{Acknowledgments}

This research was partially supported
by MEXT KAKENHI \#20500008, \#18500011, \#20300010,
and by the Kayamori Foundation of Informational Science Advancement.


\newpage
\profile{KUSAKARI Keiichirou}{%
received B.E. from Tokyo Institute of Technology in 1994,
received M.E. and the Ph.D. degree
from Japan Advanced Institute of Science and Technology
in 1996 and 2000.
From 2000,
he was a research associate at Tohoku University.
He transferred to
Nagoya University's Graduate School of Information Science
in 2003
as an assistant professor
and became an associate professor in 2006.
His research interests include term rewriting systems,
program theory, and automated theorem proving.
He is a member of IPSJ and JSSST\@.}

\profile{ISOGAI Yasuo}{%
received the B.E. and M.E. degrees
from Nagoya University in 2006 and 2008, respectively.
He engaged in research on term rewriting systems.
He is going to work at Hitachi Ltd. from April 2008.
}

\profile{SAKAI Masahiko}{%
completed graduate course of Nagoya University in 1989
and became Assistant Professor,
where he obtained a D.E. degree in 1992.
From April 1993 to March 1997,
he was Associate Professor in JAIST.
In 1996 he stayed at SUNY at Stony Brook for six months
as Visiting Research Professor.
From April 1997,
he was Associate Professor in Nagoya University.
Since December 2002,
he has been Professor.
He is interested in term rewriting system,
verification of specification and software generation.
He received the Best Paper Award from IEICE in 1992.
He is a member of JSSST\@.}

\profile{Fr\'{e}d\'{e}ric Blanqui}{%
received his PhD degree in September 2001
at the University of Paris 11 (Orsay, France).
He did a postdoc at Cambridge University (UK) 
from October 2001 to August 2002,
and at Ecole Polytechnique (Palaiseau, France)
from September 2002 to August 2003.
Since October 2003,
he is permanent full-time INRIA researcher at LORIA (Nancy, France).
He is interested in rewriting theory,
type theory, termination, functional programming and proof assistants.
He received the Kleene Award for the best student paper at LICS'01,
and the French SPECIF 2001 Award for his PhD.
}


\begin{thebibliography}{99}

\bibitem{AG00}
Arts,T. and Giesl,J.,
Termination of Term Rewriting Using Dependency Pairs,
{\em Theoretical Computer Science}, Vol.236, pp.133--178, 2000.

\bibitem{B00}
Blanqui,F., 
Termination and Confluence of Higher-Order Rewrite Systems,
In {\em Proc. of the 11th Int. Conf. on Rewriting Techniques and Applications},
LNCS 1833 (RTA2000), pp.47--61, 2000.

\bibitem{BJO02}
Blanqui,F., Jouannaud,J.-P., and Okada,M.,
Inductive-data-type Systems,
{\em Theoretical Computer Science},
Vol.272, pp.41--68, 2002.

\bibitem{B06}
Blanqui,F., 
Higher-Order Dependency Pairs,
In {\em Proc. of 8th Int. Workshop on Termination} (WST2006),
pp.22--26, 2006.

\bibitem{B07}
Blanqui,F., 
Computability Closure: Ten Years Later,
In {\em Essays Dedicated to 
Jean-Pierre Jouannaud on the Occasion of His 60th Birthday},
LNCS 4600 (Rewriting, Computation and Proof), pp.68--88, 2007.

\bibitem{D82}
Dershowitz,N., 
Orderings for Term-Rewriting Systems,
{\em Theoretical Computer Science},
Vol.17(3), pp.270--301, 1982.

\bibitem{G72}
Girard,J.-Y.,
Interpr\'{e}tation fonctionnelle et \'{e}limination
des coupures de l'arithm\'{e}tique d'ordre sup\'{e}rieur.
{\em Ph.D. thesis, University of Paris VII}, 1972.

\bibitem{HM04}
Hirokawa,N., and Middeldorp,A.,
Dependency Pairs Revisited,
In {\em Proc. of the 15th Int. Conf. on Rewriting Techniques and Applications},
LNCS 3091 (RTA04), pp.249--268, 2004.

\bibitem{JO91}
Jouannaud,J.-P., Okada,M.,
A Computation Model
for Executable Higher-Order Algebraic Specification Languages,
In {\em Proc. of the 6th IEEE Symposium on Logic in Computer Science},
pp.350--361, 1991.

\bibitem{JO97}
Jouannaud,J.-P., Okada,M.,
Abstract Data Type Systems,
{\em Theoretical Computer Science},
Vol.173, No.2, pp.349--391, 1997.

\bibitem{KNT99}
Kusakari,K., Nakamura,M., and Toyama,Y.,
Argument Filtering Transformation,
In {\em Proc. of Int. Conf.
on Principles and Practice of Declarative Programming},
LNCS 1702 (PPDP'99), pp.47--61, 1999.

\bibitem{K01}
Kusakari,K.,
On Proving Termination of Term Rewriting Systems with Higher-Order Variables,
{\em IPSJ Transactions on Programming},
Vol.42, No.SIG 7 (PRO 11), pp.35--45, 2001. 

\bibitem{KS07}
Kusakari,K. and Sakai,M.,
Enhancing Dependency Pair Method
using Strong Computability in Simply-Typed Term Rewriting Systems,
{\em Applicable Algebra in Engineering, Communication and Computing},
Vol.18, No.5, pp.407--431, 2007.

\bibitem{KS}
Kusakari,K. and Sakai,M.,
Static Dependency Pair Method for Simply-Typed Term Rewriting and Related Techniques,
{\em IEICE Transactions on Information and Systems},
Vol.E92-D, No.2, pp.235--247, 2009.

\bibitem{MN98}
Mayr,R., Nipkow,N.,
Higher-Order Rewrite Systems and their Confluence,
{\em Theoretical Computer Science},
Vol.192, No.2, pp.3--29, 1998.

\bibitem{M02}
Middeldorp,A.,
Approximations for strategies and termination,
In {\em Proc. of the 2nd Int. Workshop
on Reduction Strategies in Rewriting and Programming},
Vol.70(6) of {\em Electronic Notes in Theoretical Computer Science}, 2002.

\bibitem{N91}
Nipkow,N., 
Higher-order Critical Pairs,
In {\em Proc. 6th Annual IEEE Symposium on Logic in Computer Science},
pp.342--349, 1991.

\bibitem{SWS01}
Sakai,M., Watanabe,Y., and Sakabe,T.,
An Extension of the Dependency Pair Method
for Proving Termination of Higher-Order Rewrite Systems,
{\em IEICE Transactions on Information and Systems},
Vol.E84-D, No.8, pp.1025--1032, 2001. 

\bibitem{SK05}
Sakai,M. and Kusakari,K.,
On Dependency Pair Method for Proving Termination
of Higher-Order Rewrite Systems,
{\em IEICE Transactions on Information and Systems},
Vol.E88-D, No.3, pp.583--593, 2005. 

\bibitem{SKSSN07}
Sakurai,T., Kusakari,K., Sakai,M., Sakabe,T., and Nishida,N.,
Usable Rules and Labeling Product-Typed Terms
for Dependency Pair Method in Simply-Typed Term Rewriting Systems,
{\em IEICE Transactions on Information and Systems},
Vol.J90-D, No.4, pp.978--989, 2007. (in Japanese)

\bibitem{T67}
Tait,T.T.,
Intensional Interpretation of Functionals of Finite Type.
{\em Journal of Symbolic Logic} 32, pp.198--212, 1967.

\bibitem{T03}
Terese,
Term Rewriting Systems,
Cambridge Tracts in Theoretical Computer Science, Vol.55,
{\em Cambridge University Press}, 2003.

\bibitem{TGS04}
Thiemann,R., Giesl,J., and Schneider-Kamp,P.,
Improved Modular Termination Proofs Using Dependency Pairs.
In: {\em Proc. of the 2nd Int. Joint Conf. on Automated Reasoning},
LNAI 3097 (IJCAR2004), pp.75--90, 2004.

\bibitem{U04}
Urbain,X.,
Modular \& Incremental Automated Termination Proofs.
{\em Journal of Automated Reasoning},
32(4) pp 315--355, 2004.

\end{thebibliography}
\end{document}